\documentclass[aip]{revtex4-1}
\usepackage{amsfonts,amsmath,amssymb,amsthm}
\usepackage{subcaption}
\usepackage[usenames]{color}
\usepackage{comment}
\usepackage{caption}
\usepackage{hyperref}

\newtheorem{theorem}{Theorem}

\newtheorem{claim}[theorem]{Claim}
\newtheorem{conjecture}{Conjecture}
\newtheorem{corollary}[theorem]{Corollary}

\newtheorem{lemma}[theorem]{Lemma}

\newtheorem{proposition}[theorem]{Proposition}

\theoremstyle{definition}

\newtheorem{proviso}[theorem]{Proviso}

\theoremstyle{remark}
\newtheorem{remark}[theorem]{Remark}

\newenvironment{sm}{\left(\begin{smallmatrix}}
                   {\end{smallmatrix}\right)}

\newcommand{\Arc}{{\rm Arc}}
\newcommand{\Area}{{\rm Area}}
\newcommand{\B}{\mathcal B}
\newcommand{\Base}{{\rm Base}}
\newcommand{\C}{\mathbb C}

\newcommand{\eps}{\varepsilon}

\newcommand{\Level}{{\rm Level}}
\newcommand{\NE}{{\rm NE}}
\newcommand{\SL}{{\rm SL}}
\newcommand{\SO}{{\rm SO}}
\newcommand{\SU}{{\rm SU}}
\newcommand{\TD}{{\rm TD}}
\newcommand{\Tr}{{\rm Tr}}
\newcommand{\Z}{{\mathbb Z}}
\renewcommand{\phi}{\varphi}

\usepackage{pgfplots}
\usetikzlibrary{calc}
\usepackage{pgfplotstable}
\usepackage{colortbl}
\usepackage{booktabs}
\usepackage{pgfplotstable}

\begin{document}

\title{Optimal Ancilla-free Pauli+V Circuits for Axial Rotations}

\author{Andreas Blass}
\affiliation{Mathematics, University of Michigan, Ann Arbor, MI, USA}
\author{Alex Bocharov}
\author{Yuri Gurevich}
\affiliation{Microsoft Research, Redmond, WA, USA}

\begin{abstract}
Recently Neil Ross and Peter Selinger analyzed the problem of approximating $z$-rotations by means of single-qubit Clifford+T circuits. Their main contribution is a deterministic-search technique which allowed them to make approximating circuits shallower.

We adapt the deterministic-search technique to the case of Pauli+V circuits and prove similar results. Because of the relative simplicity of the Pauli+V framework, we use much simpler geometric methods.
\end{abstract}

\maketitle

\section{Introduction}
\label{sec:intro}

For any universal basis $\B$ for single-qubit circuits, this natural problem arises: Given a single-qubit gate $G$ and real $\eps>0$, construct an ancilla-free $\B$-circuit that approximates $G$ with precision $\eps$. There is a well-known elementary reduction of this problem to its special case where $G$ is a $z$-rotation 
$R_z(\theta)= \begin{sm}e^{-i\theta/2}&0\\ 0&e^{i\theta/2}\end{sm}$ [Ref.~\onlinecite{QCQI}]. The reduction does not work for all bases but it works for Clifford+T, for Pauli+V and for most other universal single-qubit bases in the literature. We restrict attention to the special case.

The matrix of the approximating $\B$-circuit has the form
$\begin{sm}u&-v^*\\ v&u^*\end{sm}$.
In many cases, including those of Clifford+T and Pauli+V, the circuit can be efficiently constructed from the matrix. The problem becomes just to find an appropriate pair $(u,v)$ of complex numbers.

In article [Ref.~\onlinecite{Selinger}], Peter Selinger introduced a randomized-search technique for finding a desired pair $(u,v)$ in the Clifford+T framework. The result was an efficient probabilistic circuit-synthesis algorithm for the Clifford+T basis. ``Under a mild hypothesis on the distribution of primes, the expected running time of the probabilistic algorithm is polynomial in $\log(1/\eps)$, and the depth of the resulting approximating circuit is $O(1/\eps)$. If the gate to be approximated is a $z$-rotation, the T-count of the approximating circuit is $4\log_2(1/\eps) + O(1)$.''

Let V be the set of 3 unitary operators
\[
  V_1 = (I + 2iX)/\sqrt{5},\quad
  V_2 = (I + 2iY)/\sqrt{5},\quad
  V_3 = (I + 2iZ)/\sqrt{5}
\]
where $I$ is the identity matrix and $X,Y,Z$ are the single-qubit Pauli matrices. The group generated by the three operators was introduced and studied in [Refs.~\onlinecite{LPSI,LPSII}].  The use of the V basis for quantum computing was initiated and studied in [Ref.~\onlinecite{HRC}].

Using the randomized-search technique, two of the present authors and Krysta Svore developed a probabilistic algorithm analogous to Selinger's for Pauli+V (and thus for Clifford+V) circuits [Ref.~\onlinecite{BGS}]. Under a conjecture on the distribution of primes, the expected running time of their algorithm is polynomial in $\log(1/\eps)$, and the depth of the resulting circuit is $4\log_5(1/\eps) + O(1)$. The conjecture is rather credible and purely number-theoretic.

Later, also in the framework of Clifford+T, Ross and Selinger replaced the randomized-search technique with an even more efficient deterministic-search technique [Ref.~\onlinecite{RoSelinger}]. Under a hypothesis on the distribution of primes, the expected running time of the new circuit-synthesis algorithm is polynomial in $\log(1/\eps)$, and the T-count of the approximating circuit is $3\log_2(1/\eps) + O(\log\log(1/\eps))$. If an oracle to factor integers is available (e.g. Shor's factoring algorithm), the approximating circuit has the minimal possible depth.

Ross and Selinger suggested that some other universal quantum bases may be amenable to a similarly optimal synthesis. We show in this paper that this suggestion is realized in the case of the Pauli+V basis. The deterministic-search technique simplifies quite substantially in the Pauli+V case.

We present two circuit-synthesis algorithms using credible number-theoretic conjectures. The first synthesis algorithm runs in expected polynomial time with respect to $\log\frac1\eps$ and constructs a Pauli+V circuit of depth at most $3\log_5\frac1\eps + O(\log\log\frac1\eps)$ that approximates the axial rotation $R_z(\theta)$ within $\eps$. The second synthesis algorithm takes an additional input parameter, namely an error tolerance $\delta>0$. It runs in time polynomial in $\log\frac1\eps$ and $\log\frac1\delta$, and it is in fact fast and practical. It returns Nil or constructs a Pauli+V circuit of depth at most $3\log_5\frac1\eps + O(\log\log\frac1\eps)$ that approximates the axial rotation $R_z(\theta)$ within $\eps$. The probability of returning Nil is at most $\delta$.

It may be important to find a similar solution in the framework of Fibonacci circuits described e.g. in [Ref.~\onlinecite{KBS}].

\medskip\noindent{\bf Related Work}\quad
After our result was pre-announced in [Ref.~\onlinecite{BSR}], Neil Ross published another confirmation that the deterministic-search technique works for the Pauli+V basis [Ref.~\onlinecite{Ross}]. Our solution is simpler, conceptually and algorithmically.

\section{Preliminaries}
\label{sec:prelim}

\subsection{Geometry}\label{sub:geometry}

Consider the complex plane $\C$. The real and imaginary axes meet at point $0$ which is the origin of the coordinate system and will be denoted $O$. Every nonzero point $z\in\C$ can be viewed as a vector from the origin $O$ to point $z$. It will be clear from the context when a point is treated as a vector.

If $u,v$ are distinct points of the plane then $[u,v]$ is the straight-line segment between $u$ and $v$, and $|u,v|$ is the length of $[u,v]$. If, in addition, points $u,v$ are on the unit circle around $O$ and are not the opposites of each other, let $\Arc(u,v)$ be the shorter arc of the unit circle between $u$ and $v$.

Given a positive real $\eps < 1$, consider a circular segment, or meniscus, $M_0 = \{u:\: 1-\eps^2 < \Re(u) \le 1\}$ of the unit disk, centered around point $1$. Given a real number $\theta$, rotate $M_0$ to the angle $-\theta/2$ around $O$; the resulting meniscus is centered around the point $e^{-i\theta/2}$ and will be denoted $M_\eps(\theta)$ so that $M_0 = M_\eps(0)$. Menisci play important role in our approximation problem.

Given a meniscus $M = M_\eps(\theta)$, centered around the point $z = e^{-i\,\theta/2}$, let $z_1,z_2$ be the two corner points of the meniscus. Let $z_0 = (z_1+z_2)/2$ and let $z_3$ be the intersection point of the tangent lines to the unit circle at $z_1$ and $z_2$. The following terminology will be useful.
\begin{itemize}
  \item The chord $[z_1,z_2]$ is the \emph{base} of $M$, and the vectors $z_1-z_2$ and $z_2 - z_1$ are the \emph{base vectors} of $M$.
  \item $\Arc(z_1,z_2)$ is the \emph{arc} of $M$.
  \item The vector $z  - z_0$ is the \emph{handle} of $M$. Note that the handle uniquely defines the meniscus.
  \item The isosceles triangle formed by points $z_1, z_2, z_3$ is the \emph{enclosing triangle} of $M$. The base of $M$ is also the base of the triangle, and $[z_0,z_3]$ is the \emph{median} of the triangle.
\end{itemize}

\begin{lemma}\label{lem:arc}
Let $M$ be a meniscus $M_\eps(\theta)$ with base $b$ and handle $h$, and let $\mu$ and $s$ be the median and one of the two equal sides of the enclosing triangle of $M$. Then
\begin{itemize}
  \item $|h| = \eps^2$ and $|b|\approx 2\eps\sqrt2$.
  \item $\Arc(M) \approx 2\eps\sqrt2$.
  \item $|\mu| \approx 2\eps^2$.
  \item $|s| \approx \eps\sqrt2$.
\end{itemize}
\end{lemma}

The approximate equalities mean that higher powers of $\eps$ are ignored.

\begin{proof}
Let the points $z,z_0,z_1,z_2,z_3$ be as above. The first claim follows from the definition of the meniscus. Let $x = \Arc(M)/2$.

To prove the second claim, note that $\cos x = |O,z_0|/|O,z_2| = |O,z_0| = 1 - \eps^2$. Since the Taylor series for $\cos x$ is $1 - x^2/2 + \dots$, we have $x \approx \eps\sqrt2$ and $\Arc(M) \approx 2\eps\sqrt2 = O(\eps)$.

Let $\delta = |z,z_3|$. We have $1-\eps^2 = \cos x = |O,z_2|/|O,z_3| = 1/(1+\delta)$, so $1+\delta = 1/(1-\eps^2) = 1 + \eps^2 + \eps^4 + \dots$ and $\delta\approx \eps^2$. Since $\mu = [z_0,z_3]$, we have $|\mu| \approx 2\eps^2$.

By definition of $s$, we have $|s|^2 = |\mu|^2 + (|b|/2)^2 \approx (|b|/2)^2$, so $|s| \approx |b|/2$.
\end{proof}

Recall that the trace distance $\TD(U,V)$ between unitary operators $U,V$ (up to phase factors) of the Hilbert space $\C^2$ is $\sqrt{1-|\Tr(UV^\dagger)|/2}$.

\begin{lemma} \label{lem:meniscus}
Let $U$ be a unitary operator
$\begin{sm}u&-v^*\\ v&u^*\end{sm}$ on $\C^2$,
$R$ be a $z$-rotation
$R_z(\theta) =
\begin{sm}e^{-i\theta/2}&0\\ 0&e^{i\theta/2}\end{sm}$,
and $M$ be the meniscus $M_\eps(\theta)$.
Then
\[
  \TD(U,R)<\eps \iff u_M\in M
\]
where
\begin{align*}
 u_M &= \begin{cases}
        u  & \text{if } \Re(ue^{i\theta/2}) \ge0 \\
        -u & \text{otherwise}
        \end{cases}\\
        \end{align*}
\end{lemma}

\begin{proof}
Note that $\Tr(UR^\dagger) = (ue^{i\theta/2}) + (ue^{i\theta/2})^* = 2\Re(ue^{i\theta/2})$.
\begin{align*}
\TD(U,R) < \eps
&\iff \sqrt{1-|\Tr(UR^\dagger)|/2} < \eps \\
&\iff 1-|\Tr(UR^\dagger)|/2        < \eps^2 \\
&\iff 2(1-\eps^2)                  < |\Tr(UR^\dagger)| \\
&\iff 1 - \eps^2                < |\Re(u e^{i\theta/2})|\\
&\iff 1 - \eps^2                < \Re(u_M e^{i\theta/2})\\
&\iff 1 - \eps^2 < \Re(u_M e^{i\theta/2}) \le 1 \\
&\iff u_M e^{i\theta/2} \in M_\eps(0) \\
&\iff u_M \in M_\eps(\theta)
\end{align*}
The penultimate equivalence uses the fact that $|u|\le1$ which is true because $u$ occurs in a unitary matrix.
\end{proof}

\begin{corollary}\label{cor:meniscus}
If $u\in M_\eps(\theta)$ and there exists a complex number $v$ satisfying the norm equation $|u|^2 + |v|^2 = 1$, then the matrix
$U = \begin{sm}u&-v^*\\ v&u^*\end{sm}$
is at trace distance $<\eps$ from the $z$-rotation
$R = R_z(\theta)$.
\end{corollary}

\begin{proof}
Since $u\in M_\eps(\theta)$, we have $ue^{i\theta/2}\in U_0(\theta)$, $\Re(u e^{i\theta/2}) \ge0$, and $u_M=u$. Now use Lemma~\ref{lem:meniscus}.
\end{proof}

\subsection{Pauli+V}\label{sub:V}

We use the same letter for a linear operator on $\C^2$ and its matrix in the standard basis.

Unitary operators
\[
  V_1 = (I + 2iX)/\sqrt{5},\quad
  V_2 = (I + 2iY)/\sqrt{5},\quad
  V_3 = (I + 2iZ)/\sqrt{5}
\]
generate a group $V$ that is dense in the special unitary group $\SU(2)$ [Ref.~\onlinecite{LPSI}]. Here $I$ is the identity matrix, and $X,Y,Z$ are the single-qubit Pauli matrices. We
have
\[
  V_1^{-1} = (I - 2iX)/\sqrt{5},\quad
  V_2^{-1} = (I - 2iY)/\sqrt{5},\quad
  V_3^{-1} = (I - 2iZ)/\sqrt{5}.
\]
The group $V$ is in fact freely generated by $V_1, V_2, V_3$. This fact appears without proof in [Ref.~\onlinecite{LPSI}]. For completeness, we prove it here; see Corollary~\ref{cor:free} below.

In [Ref.~\onlinecite{BGS}], we explored a slightly larger group $W$ generated by the three operators $V_j$ and three Pauli operators $X, Y, Z$. Because of $Y(I+2iX) = (I-2iX)Y$ and similar relations, every product of operators in $W$ easily reduces to a normal form
\begin{equation}\label{normal}
\begin{aligned}
A_1 A_2 \cdots A_t B\quad \text{where}\quad&
 \sqrt5A_i\in \{1\pm2iX, 1\pm2iY, 1\pm2iZ\},\\
 &\ A_1 A_2 \cdots A_t
  \text{ is reduced, and}\\
 &\ B\in \{\pm I,\pm X, \pm Y, \pm Z\};
\end{aligned}
\end{equation}
in the process the number of V factors does not change but the number of Pauli factors becomes $\le1$. The product $A_1 A_2 \cdots A_t$ is reduced in the sense that no $A_{j+1}$ is the inverse of $A_j$.

By Theorem 1 in [Ref.~\onlinecite{BGS}], an $\SU(2)$ operator is in $W$ if and only if it can be given in the form
\begin{equation}\label{exact}
  \frac{1}{(\sqrt{5})^{t}}
  \left(\begin{matrix}
    u & -v^* \\
    v &  u^*
  \end{matrix}\right)
 \end{equation}
where $u,v$ are Gaussian integers.

The next theorem will relate the exponent $t$ in this formula to the number of factors in the normal form of an element of $W$.  It will also provide similar information about the images in $\SO(3)$ of the Pauli+V matrices.

Recall how matrices in $\SU(2)$ act as rotations on
three-dimensional Euclidean space.  They act by conjugation on the 3-dimensional vector space of traceless Hermitian matrices
\[
xX+yY+zZ=
\begin{pmatrix}
  z&x-iy\\x+iy&-z
\end{pmatrix},
\]
where $X$, $Y$, and $Z$ are the Pauli matrices, and where we regard the real numbers $x,y,z$ as coordinates in $\mathbb R^3$. Furthermore, this conjugation action preserves the Euclidean norm
\[
x^2+y^2+z^2=-\det\begin{pmatrix}
  z&x-iy\\x+iy&-z
\end{pmatrix},
\]
so we get a homomorphism of $\SU(2)$ into the orthogonal group $\text{O}(3)$. Because $\SU(2)$ is connected, the homomorphism actually maps into $\SO(3)$.

Under this homomorphism, the $V$-matrices correspond to rotations by $\arccos(-3/5)$ about the three coordinate axes, namely
\[
R_1=
\begin{pmatrix}
    1&0&0\\0&-\frac35&\frac45\\0&-\frac45&-\frac35
\end{pmatrix},\quad
R_2=
\begin{pmatrix}
  -\frac35&0&-\frac45\\0&1&0\\\frac45&0&-\frac35
\end{pmatrix},\quad
R_3=
\begin{pmatrix}
  -\frac35&\frac45&0\\-\frac45&-\frac35&0\\0&0&1
\end{pmatrix}.
\]

\begin{theorem}     \label{thm:free}
  \begin{enumerate}
  \item Any matrix obtained as a reduced product of $t$ factors taken from $\{{R_1}^{\pm1},{R_2}^{\pm1},{R_3}^{\pm1}\}$ has at least one entry which, when written as a fraction in lowest terms, has denominator $5^t$.
  \item Any matrix obtained as a reduced product of $t$ factors taken from $\{{V_1}^{\pm1},{V_2}^{\pm1},{V_3}^{\pm1}\}$ has the form \eqref{exact}, and it cannot be written in that form with $t$ replaced by a smaller exponent.
  \item Any matrix in the normal form described above, with $t$ factors taken from $\{{V_1}^{\pm1},{V_2}^{\pm1},{V_3}^{\pm1}\}$    followed by one factor from $\{\pm I,\pm X,\pm Y,\pm Z\}$, has the form \eqref{exact}, and it cannot be written in that form with     $t$ replaced by a smaller exponent.
  \end{enumerate}
\end{theorem}

\begin{proof}
The proof of item (1) in the theorem is rather long and is therefore given in an appendix.  We give here the easy deductions of items~(2) and (3) from item~(1).

To prove (2), consider any product $M=A_1\cdots A_t$ of $t$ factors, each of which is in $\{{V_1}^{\pm1},{V_2}^{\pm1},{V_3}^{\pm1}\}$.  Each factor is thus a matrix of Gaussian integers divided by $\sqrt 5$, so the product $M$ is a matrix of Gaussian integers dvided by $\sqrt{5^t}$.  We need to show that no lesser power of $\sqrt 5 $ can serve as the denominator for $M$.  So suppose, toward a contradiction, that $M$ is a matrix of Gaussian integers divided by $\sqrt{5^r}$ with $r<t$.  Then the the same is true of the conjugate transpose of $M$, which is also the inverse of $M$ because $M$ is unitary.  Thus, when $M$ acts by conjugation on the three-dimensional space of traceless Hermitian matrices, the denominators are (at most) $5^r$, namely a factor $\sqrt{5^r}$ from $M$ and another factor $\sqrt{5^r}$ from $M^{-1}$.  This means that the image of $M$ in $\SO(3)$, which is given by this conjugation action on traceless Hermitian matrices, involves denominators only $5^r$.  But this element of $\SO(3)$ is obtained by multiplying the $R$ matrices corresponding to the $V$ matrices that produced $M$.  So we would have a reduced product of $t$ factors from $\{{R_1}^{\pm1},{R_2}^{\pm1},{R_3}^{\pm1}\}$ with only $5^r$ in the denominator.  This contradicts item~(1).

Finally, for item~(3), we must show that what we just proved about products of the form $A_1\cdots A_t$ is also valid for $A_1\cdots A_tB$ where $B\in\{\pm I,\pm X,\pm Y, \pm Z\}$.  But this is easy, since the entries of $B$ (and of $B^{-1}$, because $B^{-1}=B$) are Gaussian integers, so multiplication by $B$ has no effect on the number of $\sqrt 5$ factors needed in the denominator.
\end{proof}

\begin{corollary}\label{cor:free}
  \begin{enumerate}
    \item The matrices $R_1, R_2,$ and $R_3$ are free generators of the subgroup of $\SO(3)$ that they generate.
    \item The matrices $V_1,V_2,$ and $V_3$ are free generators of the subgroup of $\SU(2)$ that they generate.
  \end{enumerate}
\end{corollary}

\begin{proof}
  Both parts follow immediately from the corresponding parts of the theorem.  The identity matrix cannot be represented by a nonempty reduced word in the given generators and their inverses, because it has no 5 or $\sqrt 5$ in the denominator.
\end{proof}

\begin{proposition}
The normal forms $A_1\cdots A_tB$ all represent distinct matrices, so
that every element of $W$ has a unique normal form.
\end{proposition}

\begin{proof}
Part~(3) of Theorem~4 immediately implies that two normal forms with distinct lengths represent distinct matrices, for they have different powers of $\sqrt5$ in their simplest forms.  So we need only consider normal forms of one weight $t$ at a time.  Fix $t$ for the rest of this proof.

Combining Part~(3) of Theorem~\ref{thm:free} with Theorem 1 in [Ref.~2], we find that every matrix in $W$  whose simplest form is \eqref{exact} (with our fixed $t$) is represented by a normal form $A_1\cdots A_tB$ (again with our fixed $t$).  To show that this representation is unique, it suffices to show that the number of such matrices equals the number of such normal forms.

To count the relevant matrices \eqref{exact}, write $u=a+bi$ and $v=c+di$, where $a,b,c,$ and $d$ are ordinary integers, and observe that the matrix \eqref{exact} is in $\SU(2)$ if and only if
\[
a^2+b^2+c^2+d^2=|u|^2+|v|^2=5^t.
\]
Thus, the number of matrices in $\SU(2)$ of the form \eqref{exact} is the number of representations of $5^t$ as a sum of four squares of integers.  The number of such matrices for which this is the simplest form is then obtained by subtracting the number of such four-square representations in which all of $a,b,c,$ and $d$ are divisible by 5.

By Jacobi's four-square theorem, every positive odd integer $n$ has $8\sum_{d|n}d$ representations as a sum of four squares of integers. In particular, for $n=5^t$, there are
\[
8(1+5+\cdots+5^t)=8\frac{5^{t+1}-1}{5-1}=2(5^{t+1}-1)
\]
representations of $5^t$ as a sum $a^2+b^2+c^2+d^2$.  As noted above, we must subtract the number of these representations in which all of $a,b,c,$ and $d$ are divisible by 5.  Dividing these four integers by 5, we obtain the representations of $5^{t-2}$ as a sum of four squares, so the number to be subtracted is $2(5^{t-1}-1)$.  Therefore, the number of matrices whose simplest form is \eqref{exact} is
\[
2(5^{t+1}-1)-2(5^{t-1}-1)=2\cdot 5^{t-1}\cdot(25-1)=48\cdot 5^{t-1}.
\]

Now, let us count the number of normal forms $A_1\cdots A_tB$.  We have 6 choices for $A_1$ (namely any of the ${V_i}^{\pm1}$), 5 choices for each subsequent $A_j$ (namely any of the ${V_i}^{\pm1}$ except the inverse of the immediately preceding $A_{j-1}$), and 8 choices for $B$.  That makes $48\cdot 5^{t-1}$ normal forms.  Since this count agrees with the count of matrices above, the proof of the proposition is complete.
\end{proof}

The V-\emph{count} of a $W$-operator $U$ is $t$ if $U$ has a normal form $A_1\cdots A_t B$. Every $W$-circuit implementing $U$ contains at least $t$ $V$-gates.

\subsection{Diophantine approximations}
\label{sub:da}

We presume that the reader is familiar with continued fractions, and we use Khinchin's book [Ref.~\onlinecite{Khinchin}] as our reference on continued fractions.

Every rational number $q$ has a unique continued-fraction representation $[a_0; a_1, a_2, \dots, a_n]$ where all $a_i$ are integers, $a_1,a_2,\dots$ are positive and $a_n > 1$. Every irrational number has a unique continued-fraction representation $[a_0; a_1, a_2, \dots]$ where all $a_i$ are integers, and $a_1,a_2,\dots$ are positive.
If $[a_0; a_1,\dots,a_k]$ is an initial segment of the continued-fraction representation of $\gamma$ then the reduced fraction $p_k/q_k$ represented by $[a_0; a_1,\dots,a_k]$ is the $k^{th}$ approximant (or convergent) of $\gamma$.

By a theorem of Dirichlet [Ref.~\onlinecite[Theorem~25]{Khinchin}], for any real number $\gamma$ and any integer $r\ge1$, there exist relatively prime integers $x$ and $y$ such that
\[
  |\gamma y - x| < 1/r
  \quad\text{and}\quad 1\le y\le r.
\]
The proof of Theorem~25 in [Ref.~\onlinecite{Khinchin}] includes this claim: If $p_k/q_k$ is the approximant of $\gamma$ such that $q_k\le r$ and either $r< q_{k+1}$ or else $\frac{p_k}{q_k}=\gamma$ then $|\gamma q_k - p_k| < 1/r$.

\begin{lemma}\label{lem:rational}
  There is a polynomial-time algorithm that, given a rational $g$ and integer $r\ge1$, computes an approximant $p_k/q_k$ of $g$ such that $|g q_k - p_k| < 1/r$ and $1\le q_k\le r$.
\end{lemma}

\begin{proof}
The desired algorithm is recursive. Let $\gamma_0 = \gamma$, $a_0 = \lfloor \gamma_0 \rfloor$, $p_0=a_0$, $q_0=1$, and suppose we computed already $\gamma_j,a_j,p_j,q_j$. If $\gamma_j -a_j < 1/r$, stop and output $p_j/q_j$. Otherwise let $\gamma_{j+1} = 1/ (\gamma_j-a_j)$, $a_{j+1} = \lfloor \gamma_{j+1} \rfloor$ and
\begin{align*}
  p_{j+1} &= a_{j+1}p_j + p_{j-1}\\
  q_{j+1} &= a_{j+1}q_j + q_{j-1}
\end{align*}
where $p_{-1}=1$ and $q_{-1}=0$.

To estimate the running time of the algorithm, use the fact that any $q_n\ge 2^{(n-1)/2}$ [Ref.~\onlinecite[Theorem~12]{Khinchin}]. If the output is $p_k/q_k$, we have $2^{(k-1)/2}\le q_k\le r$ and $k \le 1 +2\log_2 r$.
\end{proof}

\begin{proviso}
  Every real number $\gamma$, used as input to an algorithm in the present paper, comes with an oracle that, given the unary notation for an integer $m\ge0$, produces the part $\sum_{n=0}^m d_n/10^n$ of the decimal notation $\sum_{n=0}^\infty d_n/10^n$ for $\gamma$, where every $d_n$ is an integer and $d_1,d_2,\dots$ are in $\{0,1,\dots,9\}$.
\end{proviso}

\begin{lemma}\label{lem:real}
  There is a polynomial-time algorithm that, given a real $\gamma$ and integer $r\ge1$, computes a reduced fraction $x/y$ such that $|\gamma y - x| < 1/r$ and $y\le 2r$.
\end{lemma}

\begin{proof}
  Use the oracle companion of $\gamma$ to compute a rational $g$ such that $|\gamma - g| \le\frac1{2r^2}$. Use the algorithm of Lemma~\ref{lem:rational} to compute a fraction $x/y$ such that $1\le y\le 2r$ and $|\gamma y -x| \le\frac1{2r}$. We have
\[
 |\gamma - \frac xy| \le |\gamma - g| + |g-\frac xy|
                     \le \frac1{2r^2} + \frac1{2ry} \le \frac1{ry}
\]
and so $|\gamma y - x| \le \frac1r$.
\end{proof}

\subsection{Sums of squares}
\label{sub:squares}

We recall some well-known facts related to the problem of representing a given (rational) integer $n\ge0$ as a sum of two squares of integers. All variables will range over the integers. For brevity, we say that $n$ is S2S if it is a sum of two squares.

Every prime number of the form $4m+1$ is S2S. Given any such prime $p$, the Rabin-Shallit algorithm finds an S2S representation $x^2+y^2$ of $p$ in expected time $O(\log p)$ [Ref.~\onlinecite[Theorem~11]{RS}].

The S2S property is multiplicative. Indeed, if $m = a^2+b^2 = |a+bi|^2$ and $n = x^2+y^2 = |x+yi|^2$ then $mn = |(a+bi)(x+yi)|^2 = (ax-by)^2 + (ay+bx)^2$.

\begin{lemma}\label{lem:S2S}
  There is an algorithm that, given the representation of any number $n$ as a product of powers of distinct primes, decides whether $n$ is S2S and, if yes, produces an S2S representation $x^2+y^2$ of $n$. The algorithm works in expected polynomial time.
\end{lemma}

\begin{proof}
A number $n$ is S2S if and only if every prime factor $q$ of $n$ of the form $4m + 3$ has an even exponent in the representation of $n$ as a product of powers of distinct primes [Ref.~\onlinecite[Theorem~366]{HW}]. This criterion allows you to decide whether $n$ is S2S or not.

If $n$ is S2S, use the Rabin-Shallit algorithm and the multiplicativity of S2S to find an S2S representation of $n$.
We illustrate this part on an example. Suppose that $n = p_1p_2^3q^2$ where $p_1,p_2$ are primes of the form $4m+1$ and $q$ is a prime of the form $4m+3$. Use the Rabin-Shallit algorithm to represent $p_1,p_2$ as sums of squares. Use the algorithm of the paragraph preceding the lemma, to represent $p_1p_2$ as $c^2+d^2$. Then $n = (cp_2q)^2 + (dp_2q)^2$.
\end{proof}

By the prime number theorem, the number $\pi(n)$ of primes $\le n$ is asymptotically equal to $\frac n{\log n}$ where $\log$ means the natural logarithm. The number of primes of the form $4m+1$ that are $\le n$ is asymptotically equal to $\frac n{2\log n}$ and thus is $\Omega(\frac{n}{\log n})$. It follows that the fraction of S2S numbers $\le n$ is $\Omega(\frac1{\log n})$.

\begin{proposition}\label{pro:R2}
  There is an algorithm that, given a positive integer $n$ of the form $4m+1$ and a positive $\delta>0$, works in time $O((\log n)\log\frac1\delta)$ and returns an S2S representation of $n$ or Nil. If $n$ is prime then,with probability $>1-\delta$,  the algorithm  returns an S2S representation of $n$.
\end{proposition}

\begin{proof}
Let $A$ be the Rabin-Shallit algorithm for finding an S2S representation of a given prime number $n=4m+1$. $A$ works in two stages. At stage~1, it solves the equation $x^2=-1$ in the field $\Z_n$. That solution is used at stage~2 to produce an S2S representation of $n$. Stage~2 is performed in linear time. Stage~1 is a while loop. At the $j^{th}$ round of the loop, $A$ randomly chooses a residue $b_j\mod n$ and computes, in linear time, the greatest common divisor $d_j(x) = x-a_j$ of polynomials $(x-b_j)^2 + 1$ and $x^{2m}-1$. With probability $1/2$, the residue $a_j+b_j$ solves $x^2=-1$; if this happens, call the round successful. The expected number of the rounds is 2. So $A$ works in expected linear time.

Let $A'$ be the modification of $A$ that takes an additional input $\delta>0$ and replaces the while loop with the following for-loop where $J = \lceil \log_2(1/\delta) \rceil$. For $j=0$ to $J$ do:
\begin{enumerate}
\item Perform one round of the $A$'s while loop.
\item If the round is successful, go to stage~2.
\item If $j<J$ then increment $j$ else stop and return Nil.
\end{enumerate}
The probability that $A'$ outputs Nil is $1/2^{J+1}<\delta$. The worst-case running time of $A'$ is $O((\log n)\log\frac1\delta)$.

The desired algorithm $A''$ is the modification of $A'$ where the input integer $n=4m+1$ is not necessarily prime. $A''$ simulates $A'$ on the given $n$ and $\delta$. If $A'$ returns an alleged S2S representation $c^2+d^2$ then $A"$ checks whether the representation is genuine and returns the same S2S representation of $n$ if it is genuine. In all other cases, $A''$ returns Nil.
\end{proof}

\section{Adjusting a meniscus} 
\label{sec:adjust}

We present an algorithm that, given a meniscus $M_\eps(\theta)$ of the unit disk, constructs an operator $\tau\in \SL(2,\Z)$ such that $\tau M$ resides in a vertical band of width $O(\eps^{3/2})$.

Call a complex number $z$  \emph{quasi-rational} if $\Re(z)=0$ or $\Im(z)/\Re(z)$ is rational.
Every nonzero quasi-rational $z$ has a unique \emph{reduced} presentation of the form $\mu(a+bi)$ where $\mu$ is real and positive and where $a,b$ are mutually prime integers; if $\Re(z)\ne0$ then $b/a$ is the reduced form of the fraction $\Im(z)/\Re(z)$.

Observe that a vector $r$ is orthogonal to a given quasi-rational vector $\mu(a+bi)$ if and only if $r$ is quasi-rational of the form $\nu(b-ai)$.

\begin{lemma}\label{lem:tau}
Let $q$ be a nonzero quasi-rational with reduced presentation $\mu(a+bi)$. For any nonzero complex number $r$ orthogonal to $q$, there is $\tau\in \SL(2,\Z)$ such that
\begin{enumerate}
  \item  $\tau q = \mu i$ and
         $|\Im(\tau r)| < |\Re(\tau r)|$,
  \item $|\tau q| = |q|/\sqrt{a^2+b^2}$,
  \item $|\Re(\tau r)| = \sqrt{a^2+b^2}\, |r|$.
\end{enumerate}
\end{lemma}

\begin{proof}[Proof of the lemma]

\smallskip\noindent
1. Use the extended Euclidean algorithm to find integers $u,v$ with $ua + vb =1$ and let $\tau_0 = \begin{sm}b&-a\\u&v\end{sm}$. Then $\tau_0q = \mu i$ and $\tau_0r$ is some complex number $\alpha+\beta i$. Since $q$ and $r$ are non-collinear, so are $\tau q$ and $\tau r$, and therefore $\alpha\ne0$.
If $|\beta|< |\alpha|$, the desired $\tau = \tau_0$. Otherwise $\tau = \tau'\tau_0$ where $\tau' = \begin{sm}1&0\\k&1\end{sm}$ and
$k = \lceil -\frac \beta\alpha \rceil = -\frac \beta\alpha + \delta$
for some $\delta$ such that $0\le \delta < 1$. We have $\tau q = \tau'(\mu i) = \mu i$, $\Re(\tau r) = \Re(\tau'(\alpha+\beta i)) = \alpha$, and $\Im(\tau r) = \Im(\tau'(\alpha + \beta i)) = k\alpha + \beta = \delta\alpha$, so that indeed $|\Im(\tau r)| < |\Re(\tau r)|$.

\smallskip\noindent
2. By the first part of Claim~1, $|\tau q| = |\mu| = |q|/\sqrt{a^2+b^2}$.

\smallskip\noindent
3. $\tau$ maps the rectangle with sides $q$ and $r$ to a parallelogram of the same area because $\tau\in\SL(2,\Z)$. The parallelogram has a vertical side $\tau q = \mu i$ and side $\tau r$ whose horizontal component is $\Re(\tau r)$, so its area is $|\mu|\cdot \Re(\tau r)$. The original rectangle had area $|q|\cdot|r| = |\mu|\sqrt{a^2+b^2}\,|r|$. Equating the two areas and cancelling $|\mu|$, we get the claim.
\end{proof}

\begin{corollary}
  Let $M$ be a meniscus with a quasi-rational base vector $q$ of reduced form $\mu(a+bi)$ and with handle $h$. There is $\tau\in \SL(2,\Z)$ such that $\tau q$ is vertical, $|\tau q| = |q|/\sqrt{a^2+b^2}$ and $|\Re(\tau h)| = \sqrt{a^2+b^2}\, |h|$.
\end{corollary}

Originally, to achieve the goal of this section, we intended to show that, for every meniscus $M$ of the unit circle, there exists a slightly bigger meniscus $L\supseteq M$ with a base vector $q$ and a handle $h$ and there exists an operator $\tau\in\SL(2,\Z)$ such that $\tau q$ is vertical and $|\Re(\tau h)| = O(\eps^{3/2})$. The intent ran into difficulties with Diophantine approximations; see \S\ref{sub:da} in this connection. Fortunately there is another way.

\begin{theorem}\label{thm:tau}
  Let $M$ be a meniscus $M_\eps(\theta)$ of the unit disk. There is $\tau\in \SL(2,\Z)$ such that $\tau M$ resides in a vertical band of width $O(\eps^{3/2})$. Moreover, there is a polynomial-time algorithm that, given $\eps$ and $\theta$, constructs the desired $\tau$.
\end{theorem}

\begin{proof}[Proof of Theorem~\ref{thm:tau}]
Since the isometry
$\begin{sm}0&-1\\ 1&0\end{sm} \in \SL(2,\Z)$
makes horizontal bands vertical, it suffices to prove the version of the theorem where ``vertical'' is replaced with ``vertical or horizontal.''

Let $z_1, z_2$ be the two corner points of $M$ and $r = \alpha + \beta i$ be the base vector $z_2-z_1$. We assume that $\alpha\beta\ne0$; otherwise we have nothing to do. Without loss of generality, $z_1$ is the left of the two corner points of $M$, so that $\alpha>0$.
The enclosing triangle $Z$ of $M$ is formed by points $z_1,z_2$ and the intersection point, call it $z_3$, of the tangent lines to the unit circle at $z_1$ and $z_2$.
Without loss of generality, $|\beta|\le|\alpha|$; otherwise, instead of making $\Base(M)$ nearly vertical, we'll make it nearly horizontal (even though this may look a bit unnatural: if the base is closer to horizontal then we adjust it to become vertical, and if it's closer to vertical then we adjust it to become horizontal.)

We are going to construct an operator $\tau\in \SL(2,\Z)$ such that the horizontal projections of all sides of the triangle $\tau Z$ are $O(\eps^{3/2})$.
Let $\gamma = \beta/\alpha$ and $n = \lceil 1/\sqrt\eps\, \rceil$. Apply the algorithm of Lemma~\ref{lem:real} to construct a reduced fraction $b/a$ such that $|\gamma a - b| < 1/n \le \sqrt\eps$ and $1\le a \le 2n = 2\lceil 1/\sqrt\eps\,\rceil$.

Since $|\gamma a - b| < 1$ and $b$ is an integer, $b$ and $\gamma a$ cannot have the opposite signs. Hence $b=0$ or $b$ has the sign of $\gamma a$ which is also the sign of $\gamma$ and $\beta$. Recall that $\alpha\ge\beta$. We claim that $a \ge|b|$. Indeed,
\[
 |b|-a  = |b| - a|\gamma| - a(1 - |\gamma|)
        \le |b| - a|\gamma|
        = |b - \gamma a| < \sqrt\eps < 1
\]
and thus $|b|\le a$.

Let $q$ be the quasi-rational $\alpha + \frac ba \alpha i = \frac\alpha a (a + bi)$. Recall from \S\ref{sub:geometry} that $|r| = \Theta(\eps)$. We have $r - q = (\beta - \frac ba\alpha) i$ and $ |r - q| = |\frac\alpha a (\gamma a - b)| \le \frac{|r|}a |\gamma a -b| = O(\eps^{3/2}/a)$. 

Consider meniscus $L = M_\delta(\eta)\supseteq M$ such that $\Base(L)$ is parallel to $q$ and touches $\Base(M)$ at $z_1$ or $z_2$. We consider only the case $z_1\in\Base(L)$; the other case is similar.
Let $y_1=z_1$ and $y_2$ be the other end of $\Base(L)$. The enclosing triangle of $L$ is formed by points $y_1,y_2$ and the intersection point, call it $y_3$, of the tangent lines to the unit circle at $y_1$ and $y_2$.

We have
$\Arc(y_1,y_2) = \Arc(z_1,z_2) + \Arc(z_2,y_2) = O(\sqrt\eps /a)$ and, by Lemma~\ref{lem:arc}, $\delta = O(\sqrt\eps /a)$.
Let $h$ be the handle of $L$. By Lemma~\ref{lem:arc}, $|h| = \delta^2$.
By Lemma~\ref{lem:tau}, there is $\tau\in \SL(2,\Z)$ that makes $\Base(L)$ vertical and such that  $|\Re(\tau h)| = \Theta(a\cdot|h|) =\Theta(a\delta^2)$.
Furthermore, a particular $\tau = \tau' \tau_0$ is constructed in the proof of the lemma; we are going to take advantage of that.

The horizontal projection of $[\tau z_1, \tau z_2]$ is $\Re(\tau r)$. We claim that the length of the horizontal projection is $O(\eps^{3/2})$. Since $\tau'$ preserves the real part of any vector, it suffices to show that $|\Re(\tau_0 r)| = O(\eps^{3/2})$. We have
\[
 \Re\left[\begin{sm}b&-a\\u&v\end{sm}
 \begin{sm}\alpha&\beta\end{sm}\right]
 = b\alpha - a\beta = \alpha(b - a\gamma).
\]
So $|\Re(\tau_0 r)| \le |r|\cdot |a\gamma - b| =
O(\eps^{3/2})$. It remains to show that the length of the horizontal projection of $[\tau z_1, \tau z_3]$ is $O(\eps^{3/2})$.

Let $z = z_3 - z_1$ and $y = y_3 - y_1$.
By Lemma~\ref{lem:arc},
$|z| \approx \eps\sqrt2$ and
$|y| \approx \delta\sqrt2$; so
$\frac{|z|} {|y|} \approx \frac\eps\delta$.
Since the vectors $y,z$ are collinear and operator $\tau$ is linear and projection operators are linear as well,
$\frac{|\Re(\tau z)|}{|\Re(\tau y)|} \approx \frac\eps\delta$. Since $\Re(\tau y) = \Re(\tau h)$, we have
\[
 |\Re(\tau z)|
 = \frac\eps\delta \Re(\tau h)
 = \Theta\left( \frac\eps\delta a\delta^2\right)
 = \Theta(a\eps\delta)
 = O(\eps^{3/2}).
\]

It remains to estimate the running time of our algorithm. To this end we need only to estimate the time needed to compute integers $a,b$ and then integers $u,v$ such that $au+bv=1$; the remaining work takes constant time. By Lemma~\ref{lem:real}, integers $a,b$ are computed in polynomial time.  The Euclidean algorithm that computes $u,v$ runs in polynomial time as well.
\end{proof}

\section{Deterministic search}
\label{sec:detsearch}

We explain the deterministic search, how it works and why. We do that essentially on the Pauli+V example. But, to simplify the exposition and minimize distractions, we abstract away some details. It would be easy to abstract away more details with the price of making the exposition a little more involved.

Consider a finite universal basis $\B$ for single-qubit gates. $\B$ may contain gates considered negligible gates; in such a case the depth of a $\B$-circuit is the number of non-negligible gates in the circuit. In the Pauli+V case, the Pauli gates may be considered negligible, because they are relatively cheap to implement and because at most one Pauli gate occurs in the normal form \eqref{normal}.

Assume that $\B$ comes with a partial function $\Level(u)$ that assigns nonnegative integers to some complex numbers and with an equation $\NE_t(u,v)$ such that the following constraints C1-C4$'$ hold.
\begin{enumerate}
\item[C1.] If $u,v$ are of levels $\le t$ and $\NE_t(u,v)$ holds, then the matrix
    $\begin{sm}u&-v^*\\ v&u^*\end{sm}$
    is unitary and exactly realizable by a $\B$ circuit of depth $\le t$.
\end{enumerate}
The ``NE'' in ``$\NE_t(u,v)$'' alludes to the fact that the condition is typically expressed as a norm equation on $u$ and $v$, with parameter $t$. Whether $\NE_t(u,v)$ is a pure norm equation or not, we  will refer to it as the \emph{norm equation}. The reader may be interested how the Pauli+V fits the general scheme, in particular what are the levels and norm equation in the Pauli+V scheme; all these questions will be addressed in the next section.

Recall that we are seeking to approximate a given $z$-rotation
$R_z(\theta) = \begin{sm}e^{-i\theta/2}&0\\ 0&e^{i\theta/2}\end{sm}$
to a given precision $\eps$. Let $M=M_\eps(\theta)$.

If $u\in M$ and $\Level(u)\le t$, call $u$ a \emph{candidate of level $\le t$}. By Corollary~\ref{cor:meniscus}, if $u$ is a candidate of level $\le t$ and $v$ is a complex number of level $\le t$ such that $\NE_t(u,v)$ holds, then the matrix
$U = \begin{sm}u&-v^*\\ v&u^*\end{sm}$
is at a distance $<\eps$ from $R_z(\theta)$. Accordingly, call a candidate $u$ of level $\le t$ a \emph{winning candidate of level $\le t$} (or simply a \emph{winning candidate} if $t$ is clear from the context) if there is a complex number $v$ of level $\le t$ such that $\NE_t(u,v)$ holds.

Let $C_t$ be the set of all candidates of levels $\le t$, and let $W_t$ be the subset of all winning candidates of levels $\le t$. Assume that the following constraints C2--C4 hold. By default, $\log$ means natural logarithm.
\begin{enumerate}
\item[C2.] There is an efficient algorithm that, given $t$, enumerates the candidates of level $t$.
\item[C3.] There exist a real $a>1$ and an integer $t_0$ such that $|C_{t_0}| >1$, and $|C_t|\le1$ for all $t<t_0$, and $|C_t| \ge a^{t-t_0}$ for all $t>t_0$ such that $t-t_0$ is even. Here $t_0$ depends on $\eps$ and $\theta$ while $a$ depends on neither.
\item[C4.] $|W_t| = \Omega(|C_t|/t)$ as $t\to\infty$, uniformly with respect to $\eps$.
\end{enumerate}

The requirement that $t-t_0$ be even in C3 reflects a peculiarity of the Pauli+V case.

\begin{lemma}\label{lem:explore}
$|W_t|\ge k$ for some $t = t_0 + c\,\log t_0$ where $c\ge0$ depends on $k$ and $\eps$ but not $\theta$.
\end{lemma}

\begin{proof}
Let $a,t_0$ be as in C3. By C4, there exist a real $b>0$ and an integer $t_1\ge1$, independent of $k$, $\eps$ or $\theta$, such that $|W_t| \ge b|C_t|/t$ for all $t\ge t_1$.

Let $t\ge t_0$, $t\ge t_1$ and $t-t_0$ be even.  Then $|W_t|\ge b|C_t|/t \ge ba^{t-t_0}/t$, and so $|W_t|\ge k$ if $a^{t-t_0} \ge kt/b$.
If $t = t_0 + x\,\log t_0$ then
\begin{equation}\label{ktest}
 |W_t|\ge k \text{ if }
 (a^{\log t_0})^x \ge \frac{kt_0}b + \frac{k\log t_0}b x.
\end{equation}
The exponential function $(a^{\log t_0})^x$ of $x$ quickly outgrows the linear function $\frac{kt_0}b + \frac{k\log t_0}b x$ of $x$. The desired $t$ is $t_0 + c\,\log t_0$ where $c$ is the least nonnegative real such that $t_0 + c\log t_0\ge t_1$, $x\log t_0$ is an even integer and the premise of the implication~\eqref{ktest} holds.
\end{proof}

Our goal is to (efficiently) find a winning candidate, preferably of low level.  Our ability to do this depends on our ability to tell whether a given candidate is winning, and in this connection we consider two scenarios.

\smallskip\noindent\emph{Scenario~1}\quad
We have a deterministic decision procedure that, given an integer $t\ge0$ and a candidate $u\in C_t$, decides in polynomial time whether $u\in W_t$. Then the following obvious deterministic search finds a winning candidate of the minimal possible level. Explore candidates of levels $0$, then candidates of level $1$, etc. until a winning candidate of some level is found. The efficiency of such a deterministic search crucially depends on the efficiency of the decision procedure.

\smallskip\noindent\emph{Scenario~2}\quad
We have a randomizing procedure that, given an integer $t\ge0$ and a candidate $u\in C_t$, decides in expected polynomial time whether $u$ belongs to a subset $W'_t$ of $W_t$ subject to the following constraint.

\begin{enumerate}
\item[C4$'$] $|W'_t| = \Omega(|C_t|/t)$ as $t\to\infty$, uniformly with respect to $\eps$.
\end{enumerate}

Then the following randomizing search finds a candidate in $W'_t$ of the minimal possible level. Explore candidates of levels $0$, then candidates of level $1$, etc. until, for some $t$, a member of $W'_t$ is found.

\begin{lemma}\label{lem:explore'}
$|W'_t|\ge k$ for some $t = t_0 + c\,\log t_0$ where $c\ge0$ depends on $k$ and $\eps$ but not $\theta$.
\end{lemma}

\begin{proof}
Just replace the reference to C4 with a reference to C4$'$ in the proof of Lemma~\ref{lem:explore}.
\end{proof}

\section{Optimal Pauli+V circuits}

\subsection{Pauli+V candidates}
\label{sub:candidates}

We specialize \S\ref{sec:detsearch} to the Pauli+V case.
All the assumptions made in \S\ref{sec:detsearch} need to be justified.

Define a complex number $u$ to be of level $\le t$ if $\sqrt{5^s}u$ is a Gaussian integer for some nonnegative integer $s\le t$. The norm equation $\NE_t(u,v)$ has a particularly simple form in the Pauli+V case: $|u|^2 + |v^2| = 1$. If $\sqrt{5^t} u, \sqrt{5^t} v$ are Gaussian integers $a+bi, c+di$ respectively then the norm equation becomes $a^2 + b^2 + c^2 + d^2 = 5^t$. By Theorem~1 in [Ref.~\onlinecite{BGS}], mentioned in \S~\ref{sub:V}, constraint C1 holds.

Toward verifying C2, construct a linear transformation $\tau$ as in Theorem~\ref{thm:tau}. The transformation $\tau$ maps straight lines into straight lines and ellipses into ellipses; it preserves areas and  convexity.
The elliptical meniscus $\tau M$ is enclosed in a vertical band of width $O(\eps^{3/2})$.

The inflated elliptical meniscus $\sqrt{5^t} \tau M$ is enclosed in a vertical band that projects onto the real segment $[l_t, r_t]$ with $r_t-l_t = O(\sqrt{5^t} \eps^{3/2})$.
$\sqrt{5^t} \tau M$ is bounded by segments \begin{align*}
  & a + i f(a):\ l_t \le a \le r_t\\
  & a + i g(a):\ l_t \le a \le r_t
\end{align*}
of a straight line and of an ellipse respectively. To simplify the exposition, we consider only the case where the straight line segment is above the ellipsis segment.

Each Gaussian integer in $\sqrt{5^t} \tau M$ belongs to a vertical segment $[n+ig(n),n+if(n)]$ where $n$ is an integer in the segment $[l_t,r_t]$; see Figure \ref{fig:menuscus:enumeration}. This allows us to enumerate efficiently all Gaussian integers $a+bi$ in $\sqrt{5^t} \tau M$ and thus to enumerate efficiently all candidates
$\frac {\tau^{-1}(a+bi)} {\sqrt{5^t}}$
of levels $\le t$.
Constraint C2 holds.

\begin{figure}[t]
\includegraphics[width=3.5in]{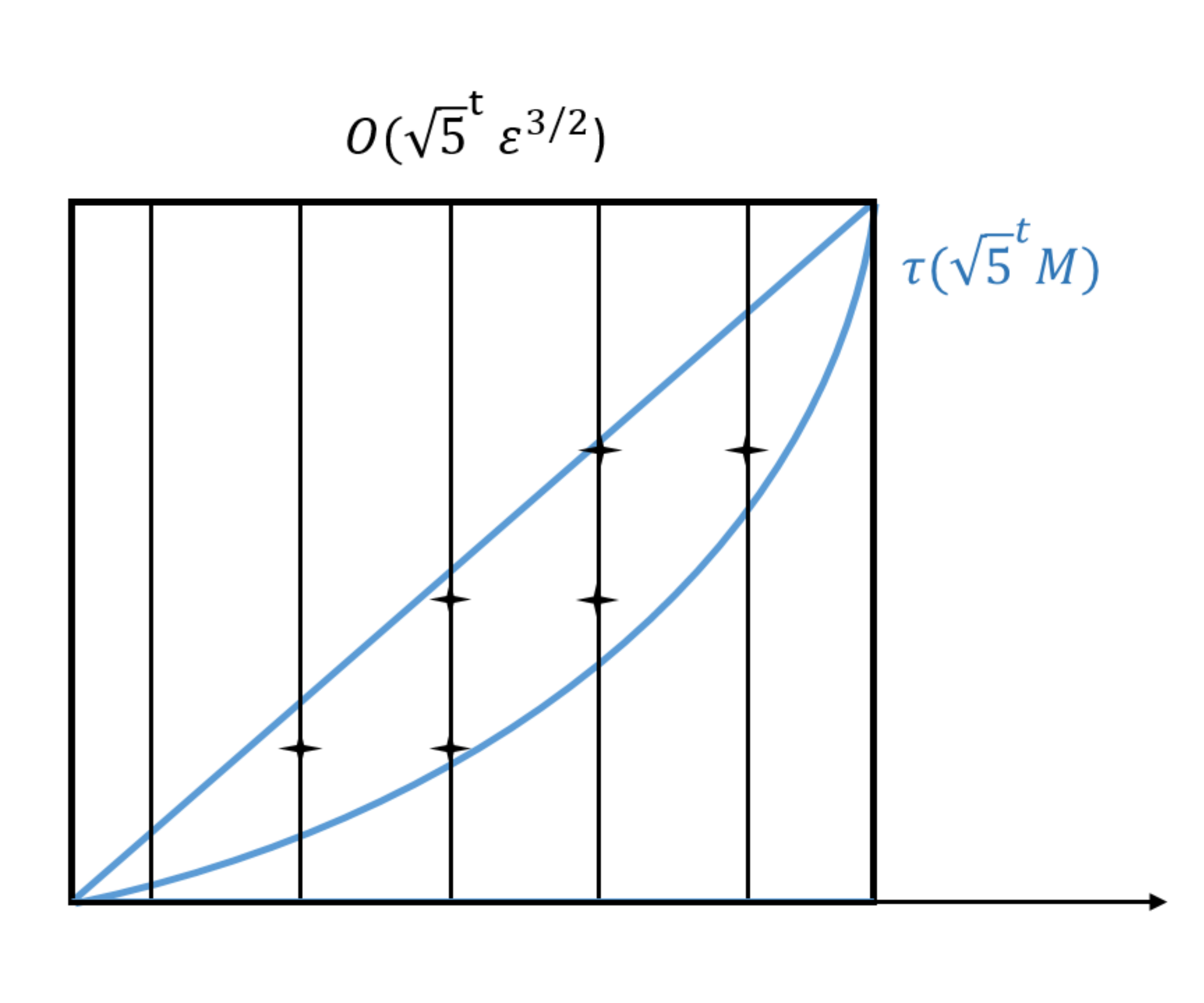} 

\caption{\label{fig:menuscus:enumeration} Enumeration of Gaussian integers in $\sqrt{5^t}\tau M$. }
\end{figure}

Constraint C3 follows from the following claim based on an observation in [Ref.~\onlinecite{RoSelinger}].

\begin{claim} \label{clm:exp_growth}
If $|C_t| \ge2$ then for any integer $k \geq 0$ we have  $|C_{t+2k}| \ge 1+5^k$.
\end{claim}

\begin{proof}
Suppose that $z_1,z_2$ are Gaussian integers and candidates
$\frac{z_1}{\sqrt{5^t}}, \frac{z_2}{\sqrt{5^t}}$
of levels $\le t$ belong to $M$. For any $k$, these candidates are also of levels $\le t+2k$. Since $M$ is convex, it contains also the following intermediate candidates of levels $\le t+2k$:
$$
 \frac{5^t-j}{5^k} z_1 + \frac j{5^k} z_2\quad
 \text{where } j=1,...,5^k-1.
$$
\end{proof}

Note that a candidate $u = \sqrt{5^{-t}}(a+bi)$ of level $\le t$ is a winning candidate of level $\le t$ if and only if $5^t-a^2-b^2$ is a sum $c^2+d^2$ of two squares. Here $a,b,c,d$ are all integers,
Constraint C4 follows from the following number-theoretic conjecture of [Ref.~\onlinecite{BGS}].
\begin{conjecture}\label{con:1}
Let $A$ be the area of the meniscus $\sqrt{5^t}M_\eps(\theta)$ of the disk of complex numbers of norm $\le\sqrt{5^t}$,
and let $S$ be the set of Gaussian integers $a+bi$ in $\sqrt{5^t}M_\eps(\theta)$ such that $5^t-a^2-b^2$ is a sum of two squares (of rational integers). Then $|S| = \Omega(A/t)$.
\end{conjecture}
Actually, instead of $\Omega$, one finds $\Theta$ in [Ref.~\onlinecite{BGS}], but only the lower bound is used there and also here.

Define $W'_t$ to be the set of members $u = \sqrt{5^{-t}}(a+bi)$ of $W_t$ such that $5^t-a^2-b^2$ is a prime of the form $4m+1$. Every such prime is a sum of two squares; see \S\ref{sub:squares}.

To justify C4$'$, we need an additional number-theoretic conjecture.
\begin{conjecture}\label{con:2}
Let $A,S$ be as in Conjecture~\ref{con:1}, and let $S'$ consists of numbers $a+bi\in S$ such that $5^t-a^2-b^2$ is a prime of the form $4m+1$. Then $|S'|=\Omega(A/t)$.
\end{conjecture}

Both conjectures were found credible by the experts in analytic number theory that we consulted. The conjectures also are supported by experimentation.
The intuition behind the conjectures is that there is no correlation between sets $S,S'$ one the one side and prime numbers on the other. As far as sets $S$ and $S'$ are concerned, the distribution of prime numbers could be random.
``It is evident that the primes are randomly distributed but, unfortunately, we don't know what `random' means'' quipped the number theorist Bob Vaughan in 1990 [Ref.~\onlinecite{MSE}].

\subsection{The first circuit-synthesis algorithm}
\label{sub:syn1}

Our first circuit-synthesis algorithm is presented in Figure \ref{fig:syn1} where $P$ is a procedure that takes a positive integer $n$ as input and returns a complex number $c+di$ with $c^2+d^2 =n$ or returns Nil. We give 2 variants of the synthesis algorithm that correspond to the two scenarios of \S\ref{sec:detsearch}. The two variants differ only in their versions $P_1, P_2$ of the procedure $P$.

\begin{figure}[hbt]
  \centering
  \makeatletter
\pgfdeclareshape{datastore}{
  \inheritsavedanchors[from=rectangle]
  \inheritanchorborder[from=rectangle]
  \inheritanchor[from=rectangle]{center}
  \inheritanchor[from=rectangle]{base}
  \inheritanchor[from=rectangle]{north}
  \inheritanchor[from=rectangle]{north east}
  \inheritanchor[from=rectangle]{east}
  \inheritanchor[from=rectangle]{south east}
  \inheritanchor[from=rectangle]{south}
  \inheritanchor[from=rectangle]{south west}
  \inheritanchor[from=rectangle]{west}
  \inheritanchor[from=rectangle]{north west}
  \backgroundpath{
    \southwest \pgf@xa=\pgf@x \pgf@ya=\pgf@y
    \northeast \pgf@xb=\pgf@x \pgf@yb=\pgf@y
    \pgfpathmoveto{\pgfpoint{\pgf@xa}{\pgf@ya}}
    \pgfpathlineto{\pgfpoint{\pgf@xb}{\pgf@ya}}
    \pgfpathmoveto{\pgfpoint{\pgf@xa}{\pgf@yb}}
    \pgfpathlineto{\pgfpoint{\pgf@xb}{\pgf@yb}}
}
}
\makeatother
\usetikzlibrary{arrows}
\begin{tikzpicture}[
  every matrix/.style={ampersand replacement=\&,column sep=1cm,row sep=0.25cm},
  sink/.style={draw,thick,rounded corners,fill=gray!20,inner sep=.1cm},
  datastore/.style={draw,very thick,shape=datastore,inner sep=.1cm},
  to/.style={->,>=stealth',shorten >=1pt,semithick},
  every node/.style={align=center}]
  \matrix{
    \& \node[datastore] (ang) {Given rotation angle $\theta$ and precision $\eps$
    }; \& \\
    \& \node[sink] (zero) {$t\gets0, out\gets None$}; \& \\
    \& \node[sink] (appr) {Compute $\tau\in \SL(2,\Z)$ as in Theorem 9}; \& \\
    \& \node[sink] (rsea) {While $out=None$}; \& \\
    \& \node[sink] (design) {$t \gets t+1;$ Explore $C_t-C_{t-1}$}; \& \\
    \& \node[sink] (synt)
    {For each $\sqrt{5^{-t}}u\in C_t-C_{t-1}$
    while $out=None$ }; \& \\
    \& \node[sink] (recur) {If $P(5^t-|u|^2)\ne Nil$ \\
    $v\gets P(5^t-|u|^2$) \\
    $out \gets \{u,v\}$  }; \& \\
    \& \node[sink] (out) {Return $\sqrt{5^{-t}} \, \begin{sm}u&-v^*\\v&u^*\end{sm} $ }; \& \\
  };
    \draw[to] (zero) --(appr);
                \draw[to] (ang) --(zero);
                \draw[to] (appr) --(rsea);
                \draw[to] (rsea) --(design);
                \draw[to] (design) --(synt);
                \draw[to] (synt) --(recur);
    \draw[to] (recur) --(out);
\end{tikzpicture}
\caption{First synthesis algorithm}
\label{fig:syn1}
\end{figure}
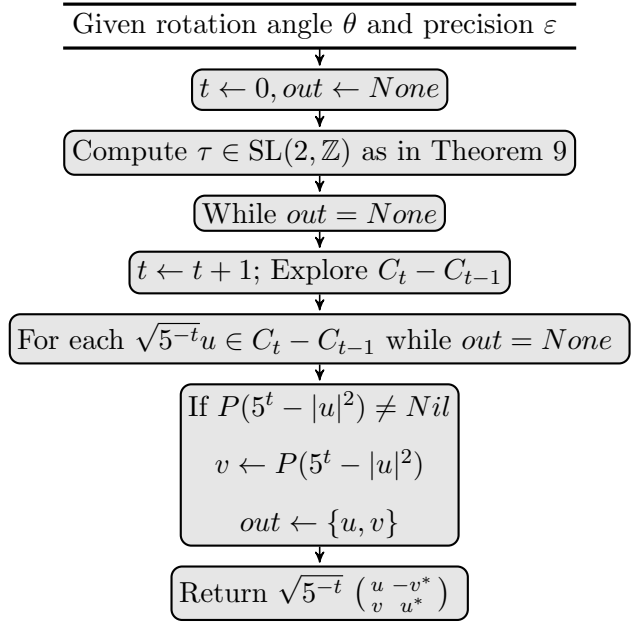

\smallskip\noindent\emph{Procedure~$P_1$}\\
This (and only this) version of $P$ presumes the availability of a quantum computer.
\begin{enumerate}
\item Use Shor's algorithm [Ref.~\onlinecite{Shor}] to factor the given number $n$. Shor's algorithm works in polynomial time on a quantum computer.
\item Return Nil if some prime factor of $n$ of the form $4m+3$ has an odd exponent in the representation of $n$ as a product of powers of distinct primes.
\item Use the algorithm of Lemma~\ref{lem:S2S} to find a representation $n=c^2+d^2$ of $n$ as a sum of two squares; return $c+d i$.
\end{enumerate}

\smallskip\noindent\emph{Procedure~$P_2$}
\begin{enumerate}
\item Check whether $n$ has the form $4m+1$, and return Nil if not.
\item Use the Lenstra-Pomerance version of the AKS primality test that runs in time $(\log n)^6\cdot (2+\log\log n)^c$ for some real constant $c$ [Ref.~\onlinecite{LP}]. (Alternatively use Miller's primality test that runs in time $O(\log n)^4$ but assumes the Extended Riemann Hypothesis [Ref.~\onlinecite{Miller}].) Return Nil if $n$ is composite.
\item Use the Rabin-Shallit algorithm to find a representation $n=c^2+d^2$ of $n$; return $c+d i$.
\end{enumerate}

The first variant of the synthesis algorithm runs in expected polynomial time and produces an approximating circuit of the least possible depth.

\begin{theorem}\label{thm:syn2}
Given a target angle $\theta$ and real $\eps>0$, the second variant of the synthesis algorithm works in expected polynomial time and constructs a Pauli+V circuit of depth at most
$3\log_5\frac1\eps + O(\log\log\frac1\eps)$ that approximates the axial rotation $R_z(\theta)$ within $\eps$.
\end{theorem}

\begin{proof}\mbox{}

\noindent\emph{Correctness}\quad
It should be obvious at this point that the algorithm produces a circuit that approximates $R_z(\theta)$ within $\eps$.

\smallskip\noindent\emph{Circuit depth}\quad
By Lemma~\ref{lem:explore'}, the synthesis algorithm produces a circuit of depth at most $t_0 + O(\log t_0)$ where $t_0$ is the least index $t$ such that $|C_t|\ge2$. It suffices to prove the following lemma.

\begin{lemma}\label{lem:t0}
$t_0 \le 3\log_5(1/\eps)$ for sufficiently small positive $\eps$.
\end{lemma}

\begin{proof}[Proof of Lemma~\ref{lem:t0}]
Let $M= M_\eps(\theta)$. By Lemma~\ref{lem:arc}, $\Area(M) \ge c\eps^3$ where $c\approx 2\sqrt2$. Let $s = 3\log_5(1/\eps)$.  Then
\[
  \Area(\sqrt{5^s}M) \ge 5^s\cdot c\eps^3
  = (1/\eps)^3 c\eps^3 = c > 2.
\]
The number $N(t)$ of Gaussian integers in $\sqrt{5^t}M$ is asymptotically equal to $\Area(\sqrt{5^t}M)$ as $t\to\infty$ [Ref.~\onlinecite{Huxley}]. So, for sufficiently small positive $\eps$, $N(s)>2$ and  $\sqrt{5^s}M$ contains at least 2 Gaussian integers.
\end{proof}

\smallskip\noindent\emph{Running time}\quad
By Lemma~\ref{lem:explore'}, $W'_t$ contains at least two elements for some $t$ of the form
$t_0 + b\log t_0$ where $b\ge0$. Accordingly, the synthesis algorithm needs to explore only candidates of levels $\le t$.
Each such candidate has the form $\frac u{\sqrt{5^t}}$ where $u$ is a Gaussian integer in $\sqrt{5^t}M$. The number of such Gaussian integers is asymptotically equal to $\Area(\sqrt{5^t}M) = 5^t\Area(M) \le 5^t4\sqrt2\eps^3$ where the inequality follows from Lemma~\ref{lem:arc}.
By Lemma~\ref{lem:t0},
$t= 3\log_5(1/\eps) + d\log_5\log(1/\eps)$
for some $d\ge0$, so that
the algorithm needs to explore at most
\[
  (1/\eps)^3 \log^d(1/\eps) 4\sqrt2 \eps^3 =
  \Theta(\log^d(1/\eps))
\]
candidates, that is only a polynomial number of candidates.

In the previous subsection, we explained how to enumerate all the candidates in $C_t$. To this end we have to walk through all the vertical segments $[n+ig(n),n+if(n)]$ where $l_t\le n\le r_t$. There are only polynomially many of those vertical segments, namely
\begin{equation}\label{half3}
\sqrt{5^t}\eps^{3/2} = (1/\eps)^{3/2} \log^{d/2}(1/\eps)
\eps^{3/2} = \log^{d/2}(1/\eps).
\end{equation}
Each vertical segment contains only polynomially many Gaussian integers, and the exploration of one candidate involves a procedure $P$ that runs in expected polynomial time. The whole algorithm runs in expected polynomial time
\end{proof}

\begin{remark}
The reader may wonder whether the transformation $\tau$ was necessary. It was. The original $M$ is enclosed in a vertical band of width that may be of the order of $\eps$. Replacing $\eps^{3/2}$ with $\eps$ gives us exponentially many vertical segments.
\end{remark}

\subsection{The second circuit-synthesis algorithm}
\label{sub:syn2}

The only essential difference of our second synthesis algorithm SA2 from the first one is its version of procedure $P$ whose input includes a positive real $\delta$ in addition to an integer. But this influences the forms of input and output of SA2. The input of SA2 comprises three components: a target angle $\theta$, a precision $\eps>0$ and an error tolerance $\delta>0$. SA2 returns either an approximation to the rotation $R_z$ or Nil. The probability of returning Nil is at most $\delta$. SA2 is presented in Figure~\ref{fig:syn2}.

\begin{figure}[hbt]
  \centering
  \makeatletter
\pgfdeclareshape{datastore}{
  \inheritsavedanchors[from=rectangle]
  \inheritanchorborder[from=rectangle]
  \inheritanchor[from=rectangle]{center}
  \inheritanchor[from=rectangle]{base}
  \inheritanchor[from=rectangle]{north}
  \inheritanchor[from=rectangle]{north east}
  \inheritanchor[from=rectangle]{east}
  \inheritanchor[from=rectangle]{south east}
  \inheritanchor[from=rectangle]{south}
  \inheritanchor[from=rectangle]{south west}
  \inheritanchor[from=rectangle]{west}
  \inheritanchor[from=rectangle]{north west}
  \backgroundpath{
    \southwest \pgf@xa=\pgf@x \pgf@ya=\pgf@y
    \northeast \pgf@xb=\pgf@x \pgf@yb=\pgf@y
    \pgfpathmoveto{\pgfpoint{\pgf@xa}{\pgf@ya}}
    \pgfpathlineto{\pgfpoint{\pgf@xb}{\pgf@ya}}
    \pgfpathmoveto{\pgfpoint{\pgf@xa}{\pgf@yb}}
    \pgfpathlineto{\pgfpoint{\pgf@xb}{\pgf@yb}}
}
}
\makeatother
\usetikzlibrary{arrows}
\begin{tikzpicture}[
  every matrix/.style={ampersand replacement=\&,column sep=1cm,row sep=0.25cm},
  sink/.style={draw,thick,rounded corners,fill=gray!20,inner sep=.1cm},
  datastore/.style={draw,very thick,shape=datastore,inner sep=.1cm},
  to/.style={->,shorten >=1pt,semithick},
  every node/.style={align=center}]
  \matrix{
    \& \node[datastore] (ang)
    {Given angle $\theta$, precision $\eps$ and error tolerance $\delta$
    }; \& \\
    \& \node[sink] (zero) {$t\gets0, out\gets None$}; \& \\
    \& \node[sink] (appr) {Compute $\tau\in \SL(2,\Z)$ as in Theorem 9}; \& \\
    \& \node[sink] (rsea) {While $out=None$}; \& \\
    \& \node[sink] (design) {$t \gets t+1;$ Explore $C_t-C_{t-1}$}; \& \\
    \& \node[sink] (synt)
    {For each $\sqrt{5^{-t}}u\in C_t-C_{t-1}$
    while $out=None$ }; \& \\
    \& \node[sink] (recur)
    {If $P(5^t-|u|^2,\delta)\ne Nil$ \\
    $v\gets P(5^t-|u|^2,\delta$) \\
    $out \gets \{u,v\}$  }; \& \\
    \& \node[sink] (out) {Return $\sqrt{5^{-t}} \, \begin{sm}u&-v^*\\v&u^*\end{sm} $ }; \& \\
  };
    \draw[to] (zero) --(appr);
                \draw[to] (ang) --(zero);
                \draw[to] (appr) --(rsea);
                \draw[to] (rsea) --(design);
                \draw[to] (design) --(synt);
                \draw[to] (synt) --(recur);
    \draw[to] (recur) --(out);
\end{tikzpicture}
\caption{Second synthesis algorithm}
\label{fig:syn2}
\end{figure}
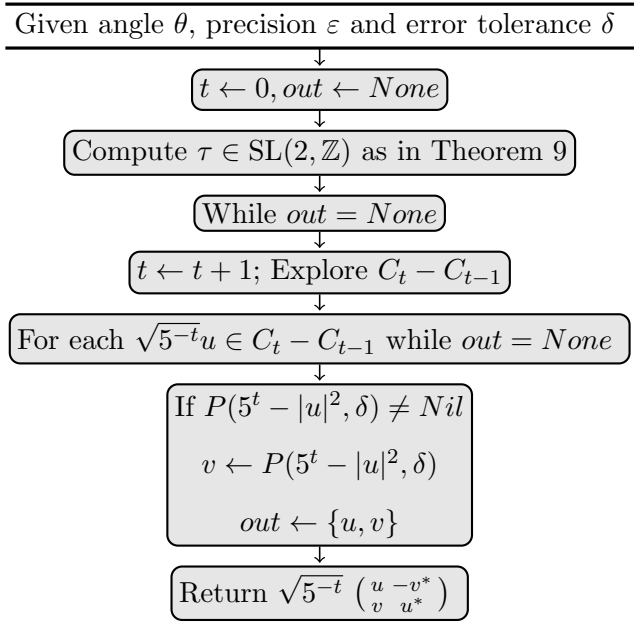

But, before we describe the new version of $P$, let's address Rabin's primality test [Ref.~\onlinecite{RabinPrime}]. It is an efficient polynomial-time primality test with a parameter $k$. If $n$ is prime then the test result is always correct. For a composite $n$ the test may declare $n$ to be prime, but the probability of such error is tiny. ``The algorithm produces and employs certain random integers $0 < b_1,\dots,b_k < n$\dots'' writes Rabin in [Ref.~\onlinecite{RabinPrime}], ``the probability that a composite number will be erroneously asserted to be prime is smaller than $\frac1{2^{2k}}$. If, say, $k = 50$, then the probability of error is at most $\frac1{2^{100}}$.''

Now we are ready to describe the new version of procedure $P$.
\begin{enumerate}
\item Check whether $n$ has the form $4m+1$, and return Nil if not.
\item Use Rabin's primality test with $k\ge \frac12  \log\frac1\delta $. Return Nil if $n$ is found to be composite.
\item Apply the algorithm of Proposition~\ref{pro:R2} to $n$ and $\delta$.
\end{enumerate}

\begin{corollary}
  The new procedure $P$ returns an S2S representation of $n$ or Nil. The probability of returning Nil is at most $\delta$.
\end{corollary}

\begin{theorem}\label{thm:syn3}
Given a target angle $\theta$ and real $\eps,\delta>0$, the second synthesis algorithm SA2 works in polynomial time and returns Nil or constructs a Pauli+V circuit of depth at most $3\log_5\frac1\eps + O(\log\log\frac1\eps)$ that approximates the axial rotation $R_z(\theta)$ within $\eps$. The probability of returning Nil is at most $\delta$.
\end{theorem}

\begin{proof}
If the primality test works as intended, SA2 works essentially like the first one. If the primality test errs, which happens with probability at most $\delta$, SA2 returns Nil or constructs a Pauli+V circuit of depth at most $3\log_5(1/\eps) + O(\log\log\frac1\eps)$ that approximates the axial rotation $R_z(\theta)$ within $\eps$.
\end{proof}

\appendix

\section{Three Free Rotations} \label{three:free:rotation}

In this appendix, we prove item~(1) of Theorem~\ref{thm:free}.  We shall be concerned with reduced words $w$ built from the formal symbols (letters) $l_1$, $l_2$, $l_3$, and their (formal) inverses. (As in the theorem, ``reduced'' means that no $l_i$ occurs next to its own inverse in $w$.) We write $w[R]$ for the product obtained by replacing each formal symbol $l_i$ by the corresponding matrix $R_i$ exhibited just before the statement of Theorem~\ref{thm:free}, and replacing ${l_i}^{-1}$ by the inverse matrix (also the transpose, as the matrices are orthogonal).  We must show that, if a reduced word $w$ has length $t$, then the corresponding matrix $w[R]$ contains at least one entry which, when written as a fraction in lowest terms, has denominator $5^t$.

We shall prove this result by induction on the length $t$ of the word. It is clearly true for $t=0$ and $t=1$.  For the induction step, we begin with some preliminary considerations to simplify the problem.

Each of the matrices $R_i$ and each of their inverses can be written as $1/5$ times an integer matrix $S_i$, namely
\[
  S_1=
  \begin{pmatrix}
    5&0&0\\0&-3&4\\0&-4&-3
  \end{pmatrix},\quad
S_2=
\begin{pmatrix}
  -3&0&-4\\0&5&0\\4&0&-3
\end{pmatrix},\quad
S_3=
\begin{pmatrix}
  -3&4&0\\-4&-3&0\\0&0&5
\end{pmatrix}.
\]
and $S_1^\top, S_2^\top,S_3^\top$.  Note that we are factoring $1/5$ out of each $R_i$ and each of the inverses $R_i^{-1}=R_i^\top$, so the remaining factors are the matrices $S_i$ and $S_i^\top$, not $S_i^{-1}$ (which differs from $S_i^\top$ by a factor 25).

Thus, if a reduced word $w$ has length $t$, then $w[R]=(1/5^t)w[S]$, where $w[S]$ is obtained from $w$ by replacing the letters $l_i$ and their inverses $l_i^{-1}$ in $w$ by the matrices $S_i$ and their transposes $S_i^\top$.  What we must prove is that, in such a product $w[S]$ of $S$ matrices, we never have all the entries divisible by 5; then, in any entry that is not divisible by 5, the overall factor of $1/5^t$ provides the denominator required for our result.

We have thus reduced our task to showing that, if we take a reduced word $w$ and substitute for each letter $l_i$ the corresponding $S_i$ and for each ${l_i}^{-1}$ the transpose $S_i^\top$ of the corresponding $S_i$, then the product matrix $w[S]$ will not have all its entries divisible by 5.

We can reduce the task further.  Since we are interested only in divisibility by 5, we can reduce all entries in the $S$ matrices, and their transposes, and their products, modulo 5.  That is, we can perform the whole calculation using matrices with entries in $\mathbb Z/5$, namely the matrices
\[
  T_1=
  \begin{pmatrix}
    0&0&0\\0&2&4\\0&1&2
  \end{pmatrix},\quad
  T_2=
\begin{pmatrix}
2&0&1\\0&0&0\\4&0&2
  \end{pmatrix},\quad
T_3=
\begin{pmatrix}
  2&4&0\\1&2&0\\0&0&0
\end{pmatrix}
\]
obtained by reducing the $S$ matrices modulo 5.  As with the $S$ matrices, we use the notation $w[T]$ for the result of taking a word $w$, replacing each $l_i$ and $l_i^{-1}$ with $T_i$ and the transpose $T_i^\top$, and multiplying the resulting matrices.  Thus, $w[T]$ is a $3\times 3$ matrix over the 5-element field $\mathbb Z/5$.

Our goal, that the product $w[S]$ of $S$ matrices obtained from a reduced word $w$ does not have all entries divisible by 5, can now be reformulated as saying that the product $w[T]$ of $T$ matrices obtained from a reduced word is not the zero matrix.

Not only are the $T$ matrices singular (because each has a column of zeros) but they have rank only 1, because, in each of them, the two non-zero columns are proportional.  The same goes for the transposes of these matrices (either by similar inspection or because transposing a matrix doesn't change its rank).  Let us write $L_i$ for the one-dimensional subspace of $(\mathbb Z/5)^3$ that is the range of $T_i$ and $L_i'$ for the range of the transpose $T_i^\top$.  Thus, $L_i$ (resp.\ $L_i'$) is generated as a vector space by either of the non-zero columns (resp.\ rows) of $T_i$.

With these preparations, we are ready for the induction step in the proof.  Suppose the claim is true for reduced words of a certain length $t\geq1$, and suppose we are given a reduced word of length $t+1$, say $QW$, where $Q$ is the first letter of our given word and $w$ is all the rest, i.e., a word of length $t$.  The first letter of $w$ (which exists as $t\geq1$) is either some $l_i$ or some $l_i^{-1}$.

Let us consider first the case that the first letter of $w$ is $l_i$. (The case of an inverse $l_i^{-1}$ will be similar.)  By induction hypothesis, the matrix $w[T]$ is not the zero matrix. Its range is included in the range $L_i$ of $T_i$ (because the range of any matrix product $JK$ is included in the range of the left factor $J$) and must therefore be all of $L_i$ (since $L_i$ has dimension only 1).  To show that $QW$ also corresponds to a non-zero matrix, it therefore suffices to show that $Q[T]$ does not annihilate $L_i$.  (As $Q$ is a single letter or inverse, $Q[T]$ is a single $T$ matrix or transpose.)  Here it is important that $Q$ is not $l_i^{-1}$, because $QW$ is a reduced word.  So we need only check that the range $L_i$ of $T_i$ is not annihilated by any of the $T_j$'s or their transposes, except for $T_i^\top$, and this is just a matter of inspection (the many 0's in the matrices make the computations trivial).

In the case where the first letter of $w $ is some $l_i^{-1}$, we similarly reduce the problem to showing that the range $L_i'$ of $T_i^\top$ is not annihilated by any of the $T_j$'s or their transposes except for $T_i$.  This again can be done by inspection, thereby completing the induction and thus the proof of Theorem~\ref{thm:free}(1).

\end{document}